\documentclass[conference]{IEEEtran}
\IEEEoverridecommandlockouts

\usepackage{cite}
\usepackage{amsmath,amssymb,amsfonts, amsthm}
\usepackage{algorithmic}
\usepackage{graphicx}
\usepackage{textcomp}
\usepackage{xcolor}
\def\BibTeX{{\rm B\kern-.05em{\sc i\kern-.025em b}\kern-.08em
    T\kern-.1667em\lower.7ex\hbox{E}\kern-.125emX}}

\newtheorem{lemma}{Lemma}
\newtheorem{theorem}{Theorem}
\newtheorem{corollary}{Corollary}
\newtheorem{remark}{Remark}

\begin{document}

\title{VLSF Decoding with Reliability Guarantees over Correlated Noncoherent Fading Channels
\thanks{
This work is supported in part by the French National Agency for Research (ANR) through the project ANR-23-CMAS-0023 of the RIS3 program and the project ANR-22-PEFT-0010 of the France 2030 program PEPR Réseaux du Futur; and in part by the Agence de l'innovation de défense (AID) through the project UK-FR 2024352.}
}

\author{%
 \IEEEauthorblockN{Guodong Sun\IEEEauthorrefmark{1}, Samir M. Perlaza\IEEEauthorrefmark{1}\IEEEauthorrefmark{2}\IEEEauthorrefmark{3}, Philippe Mary\IEEEauthorrefmark{4}, Jean-Marie Gorce\IEEEauthorrefmark{5}}

 \IEEEauthorblockA{ Emails: \{guodong.sun, samir.perlaza, jean-marie.gorce\}@inria.fr,  philippe.mary@insa-rennes.fr
             }
 
 \IEEEauthorblockA{\IEEEauthorrefmark{1} Centre Inria d'Université Côte d'Azur, Inria, Sophia Antipolis, France
             }
     \IEEEauthorblockA{\IEEEauthorrefmark{2} Laboratoire GAATI, Université de la Polyn\'{e}sie fran\c{c}aise, Fa`a`\={a}, French Polynesia
             } \IEEEauthorblockA{\IEEEauthorrefmark{3} ECE Dept, Princeton University, Princeton, 08544 NJ, USA
             }         
  \IEEEauthorblockA{\IEEEauthorrefmark{4} Univ. Rennes, INSA, CNRS, IETR UMR 6164 F-35000, Rennes, France    }  
                 	
  \IEEEauthorblockA{\IEEEauthorrefmark{5} Inria, INSA Lyon, CITI Laboratory, UR3720, Villeurbanne, France
             }
}

\maketitle

\begin{abstract}
This paper studies reliability-guaranteed decoding for variable-length stop-feedback (VLSF) codes over correlated noncoherent fading channels. 
The decoding rule is based on the evolution of the information density associated with a given channel input-output realization.
Due to channel memory, exact evaluation of this information density is intractable. 
To enable constructive decoding, computable finite-blocklength lower and upper bounds on the information density that hold uniformly over time along each input-output sequence are derived.
The lower bound enables a stopping-time analysis for VLSF decoding and has an operational meaning, while the upper bound { provides a reference for} the relaxation gap, which is explicitly characterized. 
{As a concrete application, the Gauss-Markov fading channel with Gaussian signaling is considered to numerically investigate the stopping-time distribution and the impact of fading correlation on decoding performance.}
\end{abstract}
\begin{IEEEkeywords}
	Variable-length stop-feedback code, noncoherent channel, time-correlated fading, H\"{o}lder inequality, R\'{e}nyi divergence.
\end{IEEEkeywords}

\section{Introduction}
Variable-length stop-feedback (VLSF) coding has been extensively studied for coherent memoryless channels, for which optimal decoding rules and sharp non-asymptotic performance characterizations are available~\cite{polyanskiy2011feedback, vakilinia2016optimizing, yang2022incremental, yavas2023variable, ostman2020short}.
At finite blocklengths, however, channel estimation incurs a non-negligible blocklength overhead, particularly in fading environments in which the channel evolves rapidly and experiences symbol-wise temporal correlation.
Such scenarios lead to noncoherent channels with fading memory, for which existing VLSF analyses do not directly apply.

For noncoherent fading channels, most established results rely on asymptotic analyses~\cite{lapidoth2003capacity, lapidoth2005asymptotic, deng2004information, dorpinghaus2012achievable}.
These works characterize the double-logarithmic growth of capacity with respect to signal-to-noise ratio (SNR). 
In particular, Lapidoth and Moser introduced the notion of fading numbers that can quantify the impact of fading memory in the high-SNR regime~\cite{lapidoth2003capacity, lapidoth2005asymptotic}.
While these results provide fundamental insights into ergodic capacity and high-SNR scaling laws, their expectation-based formulation makes them unsuitable for constructing finite blocklength decoding rules or stop criteria for variable-length codes. 

In the finite-blocklength regime, performance is characterized by the information density induced by a given channel input and output realization~\cite{koga2013information, verdu1994general}.
In particular, stopping rules for variable-length decoding depend on the cumulative information density evaluated along such an input-output sequence~\cite{polyanskiy2011feedback}.
To maintain analytical tractability, most existing non-asymptotic studies of noncoherent channels focus on fixed blocklength models with block-fading~\cite{qi2025noncoherent} or quasi-static fading~\cite{yang2014quasi}. 
These assumptions simplify temporal correlation by assuming blockwise-constant fading without symbol-wise channel correlation. 
As a result, a finite-blocklength analysis of information density evolution that holds uniformly over time for each input-output realization remains missing for time-correlated noncoherent channels. 
This is analogous to the asymptotic case investigated in \cite{lapidoth2003capacity, lapidoth2005asymptotic}.

From an information-theoretical perspective, short packet communication over fast fading channels with temporal correlation presents two fundamental challenges: the non-negligible blocklength cost of channel estimation; and the need for reliability guarantees in the non-asymptotic regime.
While VLSF codes address the latter in coherent, memoryless settings, their extension to noncoherent channels with memory is nontrivial due to the intractability of the information density evaluation.

This paper addresses these challenges, and the main contributions are summarized as follows:
\begin{itemize}
	\item A computable finite-blocklength lower bound on the information density associated with a given input-output realization is derived, which holds uniformly over time. The bound enables a reliability-guaranteed stopping rule for VLSF decoding over noncoherent fading channels with memory.
	\item The proposed lower bound admits a structural decomposition into a memoryless noncoherent envelope term and a penalty term capturing the effect of relaxing the channel correlation.  The penalty term is formulated as a change-of-measure cost, which is independent of the channel input. A complementary upper bound is also derived to quantify the relaxation gap.
    \item {Specializing these results to the case of a Gauss-Markov channel with Gaussian signaling}, we numerically evaluate the stopping-time distribution and investigate the impact of channel correlation on VLSF decoding performance.
\end{itemize}

The paper is organized as follows.
Section~\ref{sec:model} introduces the noncoherent channel model with memory and formulates the problem.
Section~\ref{sec:main_result} derives computable uniform lower and upper bounds on the information density induced by the channel input-output pair.
Section~\ref{sec:numerical} presents numerical results illustrating the behavior of the proposed bound and its applications.
Finally, section~\ref{sec:conclusion} concludes the paper and discusses possible future work. 

\section{VLSF Decoder and Problem Formulation}\label{sec:model}

Consider a discrete-time real-valued fading channel\footnote{ While practical fading processes are inherently complex-valued, for the ease of presentation, the main results are presented in the case of real-valued channels.} with additive Gaussian noise. 
At each time index $k\in\mathbb{N}_{>0}$, the channel input $X_k\in \mathcal{X}$ induces an output $Y_k \in \mathcal{Y}$, with $\mathcal{X}, \mathcal{Y}\subseteq \mathbb{R}$, according to 
\begin{equation}\label{eq:channel}
Y_k = H_k X_k + Z_k,
\end{equation}
where $Z_k \sim \mathcal{N}(0, \sigma_z^2)$, with $ k\in\mathbb{N}_{>0}$, are independent and identically distributed (i.i.d.). 
Throughout this work, $\mathcal{N}(\mu, \sigma^2)$ denotes the Gaussian probability measure with mean $\mu$ and variance $\sigma^2$.
The fading process $\{H_k\}_{k>0}$ is modeled as a first-order Gauss-Markov process
\begin{equation}\label{eq:gauss-markov}
H_k = \rho H_{k-1} + \sqrt{1 - \rho^2} W_k,
\end{equation}
where $\rho\in [0, 1]$ is the temporal correlation coefficient.  $W_k\sim \mathcal{N}(0, 1)$, with $ k\in\mathbb{N}_{>0}$, are i.i.d.
The initialization satisfies $H_0\sim\mathcal{N}(0,1)$, independent of $\{W_k\}_{k>0}$, which ensures the fading process to be stationary with marginal distribution $H_k\sim \mathcal{N}(0, 1)$, with $ k\in\mathbb{N}_{>0}$.
This work considers a noncoherent setting in which neither the transmitter nor the receiver has instantaneous channel state information (CSI). 
Instead, both know the distribution of the Gauss-Markov fading process in \eqref{eq:gauss-markov}.

For any blocklength $n\in \mathbb{N}_{>0}$, the fading vector $H^n\triangleq (H_1, \dots, H_n)$ is jointly Gaussian with distribution
\begin{equation}\label{eq:def_PH}
	P_{H^n}=\mathcal{N}(0, \Sigma_{H^n}),
\end{equation}
where $\Sigma_{H^n}$ is a Toeplitz matrix with entries $\big(\Sigma_{H^n}\big)_{ij}=\rho^{|i-j|}$. The matrix is positive semidefinite for $\rho \in [0, 1]$.


For any $n\in \mathbb{N}_{>0}$, denote the input and output sequences by
\begin{equation}
	x^n\triangleq (x_1, \dots, x_n)~\text{and}~y^n\triangleq (y_1, \dots, y_n). 
\end{equation}
Assume a random coding ensemble in which every symbol of the input vector $X^n$ is generated i.i.d. according to $P_X$. 
The codebook $\mathcal{C}\triangleq \big\{x^{\infty}_m: m\in \{1, \dots, M\}\big\}$ consists of $M$ codewords, in which each codeword $x^{\infty}_m\triangleq (x_{m,1}, \dots)$ is drawn independently according to the product measure $P_X^{\infty}\triangleq \prod_{k=1}^{\infty}P_X$.
For each codeword $x^\infty_m$, denote the prefix by $x_m^n\triangleq (x_{m,1}, \dots, x_{m,n})$.
Consider a message $W$ is uniformly chosen over $\{1, \dots, M\}$, and when $W=m$ is transmitted, the channel input at time $k$ is $X_k=x_{m,k}$.

For a fixed input sequence $x^n$ and fading realization $h^n$, the conditional distribution of the channel output sequence $Y^n$ follows from \eqref{eq:channel} and is given by
\begin{equation}
    P_{Y^n|X^n=x^n,H^n=h^n} = \mathcal{N}(\text{diag}(x^n)h^n, \sigma^2_z I_n),
\end{equation}
where $I_n$ denotes the identity matrix of size $n$. Averaging over the jointly Gaussian fading vector $H^n$, which is independent of $X^n$, the conditional output distribution given the input sequence $x^n$ is
\begin{equation}\label{eq:numerator}
	P_{Y^n|X^n=x^n} = \mathcal{N}\Big(0, \text{diag}(x^n) \Sigma_{H^n}\text{diag}(x^n)+\sigma^2_z I_n\Big).
\end{equation}
The marginal output distribution $P_{Y^n}$ is obtained by averaging over the random coding input and channel fading. 
Let $\mathcal B(\mathcal{Y}^n)$ denote the Borel $\sigma$-algebra on $\mathcal{Y}^n$. For any measurable set $A\in\mathcal B(\mathcal Y^n)$,
\begin{align}
    P_{Y^n}(A) =& \int \int P_{Y^n|X^n=x^n,H^n=h^n}(A) \mathrm{d}P_{X^n}(x^n) \mathrm{d}P_{H^n}(h^n) \label{eq:intractability}\\
    =&  \int P_{Y^n|H^n=h^n}(A) \mathrm{d}P_{H^n}(h^n).\label{eq:denominator}
\end{align}
Both measures in \eqref{eq:numerator} and \eqref{eq:denominator} admit everywhere positive densities on $\mathbb{R}^n$ for all $n>0$, hence $P_{Y^n|X^n=x^n}\ll P_{Y^n}$.
The information density associated with a given input-output realization $(x^n, y^n)$ is defined as 
\begin{equation}\label{eq:def_information_density}
	\imath(x^n;y^n) \triangleq \log \frac{\mathrm{d}P_{Y^n|X^n=x^n}}{\mathrm{d} P_{Y^n}}(y^n).
\end{equation}

Consider a threshold-based VLSF decoding rule with a guaranteed reliability $1-\epsilon$, for some $\epsilon>0$. 
The decoder sequentially observes the channel outputs $\{y_k\}_{k>0}$ and for each codeword $m$ computes $\imath(x_m^n;y^n)$ for every $n\in \mathbb{N}_{>0}$.
The stopping time is defined as
\begin{equation}
    \tau = \inf\big\{n\in \mathbb{N}_{>0}: \exists \bar{m}\in\{1, \dots, M\}, \imath(x_{\bar{m}}^n;y^n)\geq \gamma \big\}.
\end{equation}
If at time $\tau$, there is a unique $\bar{m}$ such that $\imath(x_{\bar{m}}^\tau;y^\tau)\geq \gamma$, the decoder outputs $\hat{W}=\bar{m}$. Otherwise, an error is declared.
The decoding error probability is defined as $\mathbb{P}[W\neq \hat{W}]$, where the probability is taken over the random codebook, message, fading process, and noise.
\begin{theorem}\label{thm:threshold_based_rule}
	Let $\epsilon>0$ and $M\in \mathbb{N}_{>0}$. 
	For any threshold $\gamma\geq\log\frac{M-1}{\epsilon}$, the threshold-based VLSF decoding rule satisfies $$\mathbb{P}[W\neq \hat{W}] \leq \epsilon.$$
\end{theorem}
\begin{proof}
See Appendix A in \cite{guodong_inria_report}.
\end{proof}
\begin{remark}
Threshold-based VLSF decoding rules have been previously analyzed for coherent memoryless channels in~\cite{polyanskiy2011feedback} and \cite{yavas2023variable}. Theorem~\ref{thm:threshold_based_rule} extends this reliability guarantee to the noncoherent
setting with memory.
\end{remark}

To evaluate the information density in \eqref{eq:def_information_density}, observe that the measure $P_{Y^n|X^n=x^n}$ in \eqref{eq:numerator} admits a closed-form density and is thus analytically tractable. 
In contrast, the measure $P_{Y^n}$ in \eqref{eq:intractability} is a high-dimensional Gaussian mixture induced jointly by the fading process and the randomness of the channel input.
In general, this mixture is analytically intractable in the presence of channel memory.
This analytical intractability motivates the development of computable upper and lower bounds on the information density in \eqref{eq:def_information_density} in the sequel.

\section{Main Results}\label{sec:main_result}

Since the exact information density induced by a given input-output realization is generally intractable, the decoding rule in Theorem~\ref{thm:threshold_based_rule} is based on computable bounds that hold uniformly over time. 
In particular, a lower bound plays a direct operational role. Whenever such a lower bound exceeds the decoding threshold, the actual information density is guaranteed to exceed the same threshold. 
Hence, the decoding reliability constraint is guaranteed. 
An upper bound plays a complementary role by quantifying the gap to the lower bound, and hence the conservativeness of the relaxation.

{ The results in Theorems \ref{thm:lower_bound} and \ref{thm:lower_bound_numerator} are applicable to general correlated fading processes, and are later applied to the Gauss-Markov model in the subsequent corollaries.}


\subsection{A lower bound on the information density}
Let $\mathcal{Q}$ be the class of probability measures on $\mathbb{R}^n$ such that, for every $Q_{H^n}\in \mathcal{Q}$, $P_{H^n}$ is absolutely continue with respect to (w.r.t.) $Q_{H^n}$, i.e. $P_{H^n}\ll Q_{H^n}$,  with $P_{H^n}$ in \eqref{eq:def_PH}. For any such $Q_{H^n}$, define the likelihood ratio
\begin{equation}\label{eq:llr_H}
	L(h^n)\triangleq \frac{\mathrm{d}P_{H^n}}{\mathrm{d}Q_{H^n}}(h^n).
\end{equation} 
Throughout this section, we assume that $P_{Y^n|H^n=h^n}$, $P_{Y^n|X^n=x^n, H^n=h^n}$, and $P_{Y^n|X^n=x^n}$ admit densities w.r.t. Lebesgue measure $\mathbb{R}^n$, denoted by $f_{Y^n|H^n}(y^n|h^n)$,  $f_{Y^n|X^n, H^n}(y^n|x^n, h^n)$, and $f_{Y^n|X^n}(y^n|x^n)$, respectively. 
The marginal output in \eqref{eq:denominator} admits a density of the form
\begin{equation}\label{eq:output_change_of_measure}
    f_{Y^n}(y^n) = \int f_{Y^n|H^n}(y^n|h^n) L(h^n) \mathrm{d} Q_{{H^n}}(h^n).
\end{equation}

The following theorem provides a lower bound on the actual information density in \eqref{eq:def_information_density}.
\begin{theorem}\label{thm:lower_bound}
Let $Q_{H^n}\in\mathcal{Q}$ be a reference measure.
Fix $r>1$, and let $s\triangleq\frac{r}{r-1}$ denote its H\"{o}lder conjugate. Then, given a pair $(x^n,y^n)\in \mathcal{X}^n\times \mathcal{Y}^n$,
\begin{equation}\label{eq:holder_measure}
\begin{aligned}
\imath(x^n;y^n)
\ge  & \log f_{Y^n|X^n}(y^n|x^n)
-\frac{r-1}{r} D_r(P_{H^n}\|Q_{H^n}) \\
&-\frac{1}{s}\log
\mathbb E_{Q_{H^n}}\left[
	f_{Y^n|H^n}(y^n|H^n)^s
\right],
\end{aligned}
\end{equation}
where $D_r(P_{H^n}\|Q_{H^n})$ is  the R\'enyi divergence of order $r$.
\end{theorem}
\begin{proof}
Since $P_{H^n}\! \ll Q_{H^n}$, the Radon-Nikodym derivative 
$\frac{\mathrm{d}P_{H^n}}{\mathrm{d}Q_{H^n}}$ exists.
Fix $r\!>\!1$ and $s\!=\! \frac{r}{r-1}$ so that $\frac{1}{r}+\frac{1}{s}=1$. 
Applying H\"{o}lder's inequality~\cite[Theorem~2.11.1]{bogachev2007measure} to \eqref{eq:output_change_of_measure} gives,  
\begin{equation}\label{eq:holder_inequality}
\begin{aligned}
	& f_{Y^n}(y^n) = \int f_{Y^n|H^n}(y^n|h^n)L(h^n) \mathrm{d}Q_{H^n}(h^n) \leq\\
	& \bigg(\! \int \!\! L(h^n)^r \mathrm{d}Q_{{H^n}}(h^n)  \!\bigg)^{\!\!\frac{1}{r}}\!\! \bigg(\!\int \!\! f_{Y^n|H^n}(y^n|h^n)^s \mathrm{d}Q_{{H^n}}(h^n)\!\bigg)^{\!\frac{1}{s}}.
\end{aligned}
\end{equation}
The first factor can be expressed in terms of the R\'{e}nyi divergence of order $r$,
\begin{equation}\label{eq:renyi_term}
\mathbb E_{Q_{H^n}}[L(h^n)^r]=
\exp\Big((r-1)D_r(P_{H^n}\|Q_{H^n})\Big),
\end{equation}
where 
\begin{equation}
D_r(P_{H^n}\|Q_{H^n})
\triangleq
\frac{1}{r-1}
\log\!
\int\!
\left(\frac{\mathrm d P_{H^n}}{\mathrm d Q_{H^n}}(h^n)\right)^r
\mathrm{d} Q_{H^n}( h^n).
\end{equation}
The information density in \eqref{eq:def_information_density} can be expressed as
\begin{equation}\label{eq:LLR_information_density}
\begin{aligned}
    \imath(x^n; y^n) =& \log\frac{\mathrm d P_{Y^n|X^n=x^n}}{\mathrm d P_{Y^n}}(y^n)\\
    =&\log  f_{Y^n|X^n}(y^n|x^n) - \log f_{Y^n}(y^n),
\end{aligned}
\end{equation}
where the last equation holds whenever $P_{Y^n|X^n=x^n}$ and $P_{Y^n}$
are absolutely continuous w.r.t. Lebesgue measure $\lambda^n$.
Combining \eqref{eq:holder_inequality}, \eqref{eq:renyi_term}, and \eqref{eq:LLR_information_density}  completes the proof.
\end{proof}

\begin{remark}
The inequality in \eqref{eq:holder_inequality} holds for every $y^n$ and for all $n$. Thus, the information density lower bound in \eqref{eq:holder_measure} applies directly to the threshold-based VSLF decoding rule.
The change-of-measure factor in \eqref{eq:renyi_term} $ \exp\big((r-1)D_r(P_{H^n}\|Q_{H^n})\big)$ depends only on $(P_{H^n}, Q_{H^n})$ and the order $r$. It is therefore a deterministic constant that does not affect the sample-path fluctuations of the information density induced by $(x^n,y^n)$, but only induces a constant offset.
\end{remark}

Below, we refer to the term $ \frac{r-1}{r}D_r(P_{H^n}\|Q_{H^n})$ as the R\'{e}nyi penalty. 
The remaining term in the H\"older inequality  $\frac{1}{s}\log
\mathbb E_{Q_{H^n}}\left[
	f_{Y^n|H^n}(y^n|H^n)^s\right]$ is called the H\"{o}lder envelope.

\subsection{Specialization to i.i.d. Gaussian reference fading and Gaussian signaling}

We adopt an i.i.d. Gaussian reference fading measure, which enables closed-form evaluation of the R\'enyi penalty and a tractable expression for the H\"older envelope in \eqref{eq:holder_measure}. 
Let $Q_{H^n}=\mathcal{N}(0,\sigma_h^2I_n)$, where $\sigma_h^2$ is a design parameter.

\begin{lemma}[R\'enyi penalty]\label{lemma:renyi_penalty}
	Let $Q_{H^n}=\mathcal{N}(0,\sigma_h^2I_n)$ and $P_{H^n}=\mathcal{N}(0,\Sigma_{H^n})$ in \eqref{eq:def_PH}. 
	The R\'{e}nyi penalty in \eqref{eq:renyi_term} admits an explicit form
    \begin{equation}\label{eq:r_th_moment_result}
		\mathbb E_{Q_{H^n}}\![L(h^n)^r]\! = \!\left(\!\frac{\sigma^{2n}_{h}}{|\Sigma_{H^n}|}\! \right)^{\frac{r}{2}}\!\!\det\!\Big(\!I_n\!+\!r\sigma^2_h(\Sigma_{H^n}^{-1}\!-\!\frac{1}{\sigma^2_h}I_n)\!\Big)^{-\frac{1}{2}}.
	\end{equation}
    Moreover, $\mathbb E_{Q_{H^n}}[L(h^n)^r]<\infty$ if 
	\begin{equation}\label{eq:feasibility}
        \sigma_h^2>\frac{r-1}{r} \frac{1+\rho}{1-\rho},
	\end{equation}
	with $\rho$ in \eqref{eq:gauss-markov}.

	If \eqref{eq:feasibility} holds, then
\begin{equation}\label{eq:szego_limit}
	\begin{aligned}
		&\lim_{n\to\infty}\frac{1}{n}
\log \mathbb E_{Q_{H^n}}\left[L(h^n)^r\right]\\&=
\frac{1}{2}
\int_{-\pi}^{\pi}
\log\frac{r\sigma_h^2 S_H(\omega)}{r\sigma_h^2-(r-1)S_H(\omega)}
\frac{\mathrm d\omega}{2\pi},
	\end{aligned}
\end{equation}
where 
\begin{equation}\label{eq:S_H_omega_def}
	S_H(\omega)=\frac{1-\rho^2}{1+\rho^2-2\rho\cos \omega}.
\end{equation}
\end{lemma}
\begin{proof}
	See Appendix B in \cite{guodong_inria_report}.
\end{proof}

We now specialize in Gaussian signaling for its tractability.
Let $\phi_v(y)\triangleq\frac{1}{\sqrt{2\pi v}}e^{-y^2/(2v)}$ denote the probability density function of a zero–mean Gaussian random variable with variance $v$. 
\begin{lemma}[H\"older envelope under Gaussian signaling]\label{lemma:envelop}
	Assume Gaussian signaling $X_k\sim \mathcal{N}(0, P_0)$  for all $k>0$, independent of the fading and noise processes. Then, the AWGN channel induces
	\begin{equation}
		P_{Y_k|H_k=h_k}=\mathcal{N}(0, \sigma_z^2+P_0h_k^2),\quad \forall k>0.
	\end{equation}
	Let $f_{Y^n|H^n}(y^n|h^n)$ denote the conditional output density.
	Then,
\begin{equation}\label{eq:prod_density}
f_{Y^n|H^n}(y^n|h^n)=\prod_{k=1}^n \phi_{v_k}(y_k).
\end{equation}
Fix $s>1$ and define $c_{s,v}\triangleq\frac{1}{\sqrt{s(2\pi v)^{s-1}}}$. Then the $s$-th power of \eqref{eq:prod_density} is
\begin{equation}\label{eq:s_power_clean}
\Big(f_{Y^n|H^n}(y^n|h^n)\Big)^s=\prod_{k=1}^n c_{s,v_k}\phi_{\frac{v_k}{s}}(y_k).
\end{equation}
Moreover, under the reference fading measure $Q_{H^n}=\mathcal{N}(0,\sigma_h^2I_n)$, the H\"{o}lder envelope term in \eqref{eq:holder_measure} factorizes as
\begin{equation}\label{eq:holder_factorized}
\begin{aligned}[b]  
\int f_{Y^n|H^n}(y^n|h^n)^s&\mathrm dQ_{H^n}(h^n)
 \\&=
\prod_{k=1}^n
\int c_{s,v_k}\phi_{\frac{v_k}{s}}(y_k)\phi_{\sigma_h^2}(h_k)\mathrm dh_k,
\end{aligned}
\end{equation}
where each scalar integral is finite and computable.
\end{lemma}
\begin{proof}
The product form \eqref{eq:prod_density} follows from the conditional independence of the channel output given $H^n$.
A Gaussian density $\phi_{v}(y)$ raised to power $s$ is given by
\begin{equation}
\begin{aligned}
	\Big(\phi_v(y)\Big)^s &= \bigg(\frac{1}{\sqrt{2\pi v}}\bigg)^s e^{-sy^2/(2v)} \\
	&= \frac{1}{\sqrt{s (2\pi v)^{s-1}}}  \cdot \bigg(\frac{1}{\sqrt{2\pi v/s}}\bigg) e^{-y^2/(\frac{2v}{s})}\\
	&= c_{s,v}\phi_{\frac{v}{s}}(y) ,
\end{aligned}
\end{equation}
where $c_{s,v}\!=\!\frac{1}{\sqrt{s(2\pi v)^{s-1}}}$.
Under $Q_{H^n}$, the independence of the components of $H^n$ yields the factorization in \eqref{eq:holder_factorized}. 
\end{proof}

The expression \eqref{eq:holder_factorized} factorizes the $ n$-dimensional integral into a product of $n$ one-dimensional Gaussian integrals. 
Consequently, the H\"older envelope term can be computed with linear complexity in $n$, making the information density lower bound numerically tractable at finite blocklength.

\begin{corollary}\label{col:lower_bound}
	Consider the noncoherent fading channel in \eqref{eq:channel} with fading process $H^n$ defined in \eqref{eq:def_PH}. Assume $P_X=\mathcal{N}(0, P_0)$, $Q_{H^n}=\mathcal{N}(0, \sigma_h^2I_n)$, a H\"older exponent $r>1$ with conjugate $s=\frac{r}{r-1}$ satisfying $\sigma^2_h\geq \frac{r-1}{r}\frac{1+\rho}{1-\rho}$ with $\rho$ in \eqref{eq:gauss-markov}.
	Then, for every input-output realization $(x^n,y^n)$, the information density satisfies
\begin{equation}
    \imath(x^n; y^n) \geq \psi(x^n, y^n).
\end{equation}
where
\begin{equation}\label{eq:def_lower_bound_phi}
\begin{aligned}
	\psi(x^n, y^n) \triangleq \log f_{Y^n|X^n}(y^n|x^n) -\frac{1}{s}\log  q_s(y^n) \\
    -  \frac{r-1}{r}D_r(P_{H^n}\|Q_{H^n}),
\end{aligned}
\end{equation}
where the computable envelope term is 
\begin{equation}
	q_s(y^n)\!\triangleq\! \prod_{k=1}^n
\!\int\! c_{s, v_k}\phi_{\frac{v_k}{s}}(y_k)\phi_{\sigma_h^2}(h_k)\mathrm dh_k,
\end{equation}
with $v_k=\sigma^2_z+P_0h_k^2$.
\end{corollary}

\begin{proof}
	The result follows by combining Theorem~\ref{thm:lower_bound} with Lemmas~\ref{lemma:renyi_penalty} and \ref{lemma:envelop}.
\end{proof}

This bound provides a stopping metric of the form $\log f_{Y^n|X^n}(y^n|x^n)-\frac{1}{s}\log q_s(y^n)$, with a deterministic penalty $\frac{r-1}{r}D_r(P_{H^n}\|Q_{H^n})$.
This structure is directly compatible with VLSF decoding rules, where decoding is triggered once the lower bound exceeds a prescribed threshold.
Note that the tightness of the lower bound is determined by the parameters $(r, \sigma_h^2)$, which control the tradeoff between the H\"{o}lder envelope and the R\'{e}nyi penalty. 

\subsection{An upper bound on information density}
The following theorem provides an upper bound
on the actual information density in \eqref{eq:def_information_density}.
\begin{theorem}\label{thm:lower_bound_numerator}
Let $Q_{H^n}\in\mathcal{Q}$, and the likelihood ratio $L(h^n)$ in \eqref{eq:llr_H}. Then, for every realization $(x^n,y^n)$,
\begin{equation}
    \imath(x^n;y^n)
\le \varphi(x^n,y^n),
\end{equation}
where 
\begin{equation}\label{eq:def_upper_bound_psi}
\begin{aligned}[b]
\varphi(x^n,y^n) \triangleq & \log f_{Y^n|X^n}(y^n|x^n) + D(P_{H^n}\| Q_{H^n})\\& -\mathbb{E}_{Q_{H^n}}\big[ \log f_{Y^n|H^n}(y^n|H^n)\big],
\end{aligned}
\end{equation}
where $D(P_{H^n}\|Q_{H^n})$ denotes the Kullback–Leibler (KL) divergence.
\end{theorem}
\begin{proof}
Taking a logarithm to \eqref{eq:output_change_of_measure} and applying Jensen’s inequality to the concave function $\log(\cdot)$ give
\begin{equation}
\begin{aligned}[b]
	& \log f_{Y^n}(y^n)\\
	 =& \log \mathbb{E}_{Q_{H^n}}\big[f_{Y^n|H^n}(y^n|H^n)L(H^n)\big] \\
	\geq &\mathbb{E}_{Q_{H^n}}\big[ \log f_{Y^n|H^n}(y^n|H^n)L(H^n) \big]\\
	= & \mathbb{E}_{Q_{H^n}}\big[ \log f_{Y^n|H^n}(y^n|H^n)\big] +  \mathbb{E}_{Q_{H^n}}\big[ \log \big(L(H^n)\big)\big]\\
	=& \mathbb{E}_{Q_{H^n}}\big[ \log  f_{Y^n|H^n}(y^n|H^n)\big] - D(P_{H^n}\| Q_{H^n}),
\end{aligned}
\end{equation}
where in the last equality the second term is $- D(P_{H^n}\| Q_{H^n})$ by definition of the KL divergence.
Substituting the bound of $\log f_{Y^n}(y^n)$ into the information density completes the proof. 
\end{proof}
\begin{remark}
    Specializing to $P_X=\mathcal{N}(0, P_0)$ and $Q_{H^n}=\mathcal{N}(0, \sigma_h^2I_n)$ enables closed-form expressions for both the envelope term $\mathbb{E}_{Q_{H^n}}\big[ \log  f_{Y^n|H^n}(y^n|H^n)\big]$ and the KL penalty $D(P_{H^n}\| Q_{H^n})$, resulting in a tractable upper bound. 
\end{remark}
Theorem~\ref{thm:lower_bound_numerator} can be viewed as an analogue of the capacity upper bound of Lapidoth and Moser in~\cite{lapidoth2003capacity} valid for every input-output $(x^n,y^n)$ and for all $n$, which is obtained via a reference output distribution induced by the true input and an auxiliary fading channel.
{ Although this bound may be loose due to the probable sub-optimality of Gaussian signaling and the reference measure,
Theorems 2 and 3 constitute a first attempt to bound the information density evolution in the non-coherent fading channel, to the best of our knowledge.}

\section{Numerical Results}\label{sec:numerical}
This section evaluates the VLSF stopping rule induced by the proposed computable bounds derived in Section~\ref{sec:main_result} under an i.i.d. Gaussian signaling $X_k\sim\mathcal{N}(0, P_0)$ and an i.i.d. Gaussian reference measure $Q_{H^n}=\mathcal{N}(0, \sigma_h^2I_n)$. 
The design parameter, i.e., the H\"older exponent $r$ and the reference variance $ \sigma_h^2$ in Corollary~\ref{col:lower_bound}, are optimized via numerical search. 

\begin{figure}
    \centering    
    \includegraphics[width=0.9\linewidth]{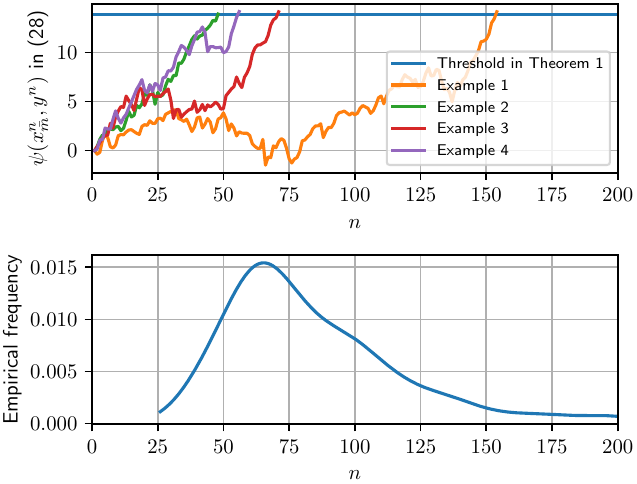}
    \caption{Top: four realizations of the information density lower bound $\psi(x_{\bar{m}}^n, y^n)$ that trigger the stopping event. Bottom: the corresponding empirical stopping-time distribution of the VLSF decoder.}
    \label{fig:evolution}
\end{figure}

Figure~\ref{fig:evolution} illustrates the behavior of the VLSF decoding rule induced by the information density lower bound $\psi(x_{\bar{m}}^n, y^n)$ in~\eqref{eq:def_lower_bound_phi}. The upper subfigure shows four representative realizations up to the threshold crossing, while the lower subfigure depicts the corresponding empirical stopping-time distribution.
The simulations are performed for a correlation coefficient $\rho=0.3$, an SNR of $\frac{P_0}{\sigma_z^2}=100$, a message size of $\log_2 M=10$ bits, and a target error probability $\epsilon=10^{-3}$, { which implies a threshold $\gamma=\log\frac{2^{10}-1}{10^{-3}} \approx 13.84$}.
The choice $\rho=0.3$ is due to the use of the i.i.d. Gaussian reference measure $Q_{H^n}=\mathcal{N}(0,\sigma_h^2 I_n)$, ensuring that the R\'{e}nyi divergence w.r.t. the true fading measure does not become excessively large.
Figure~\ref{fig:evolution} illustrates the advantage of VLSF codes: transmission terminates once the accumulated information density exceeds the decision threshold, whereas fixed blocklength codes must continue until the end of the block. 
Moreover, the stopping-time distribution characterizes the decoding delay, enabling the design of time-sensitive services with probabilistic delay guarantees under reliability constraints.

\begin{figure}
    \centering
\includegraphics[width=0.9\linewidth]{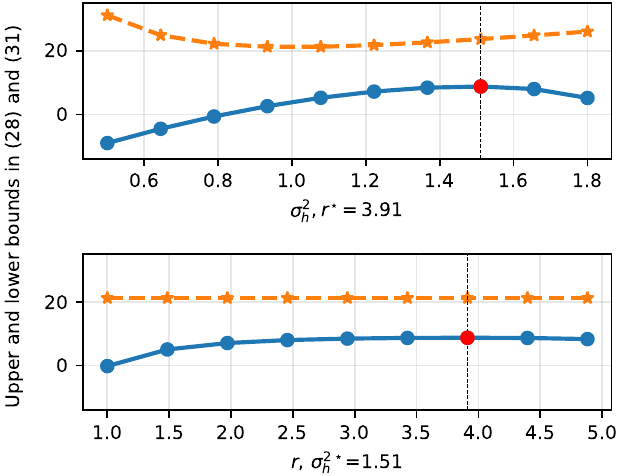}
    \caption{For $n=50$, the lower bound $\psi(x^n,y^n)$ in \eqref{eq:def_lower_bound_phi} from Corollary~\ref{col:lower_bound} is shown by the blue solid curves, while the upper bound $\varphi(x^n,y^n)$ in \eqref{eq:def_upper_bound_psi} is shown by the orange dashed curves. The optimal lower bound, obtained via a search over the free parameters $(r, \sigma_h^2)$, is highlighted in red.}
    \label{fig:figure_components}
\end{figure}

Figure~\ref{fig:figure_components} presents the average bounds on the information density for a blocklength of $n=50$ under the above setting. 
This analysis is motivated by the dependence of the lower bound in Corollary~\ref{col:lower_bound} on the design parameters, namely the R\'{e}nyi order $r$ and the reference variance $\sigma_h^2$, as well as the dependence of the upper bound in Theorem~\ref{thm:lower_bound_numerator} on~$\sigma_h^2$. 
As shown, the tightness of both bounds is determined by appropriately optimizing over the parameter pair~$(r,\sigma_h^2)$, which is performed via an exhaustive search. 
The gap between the upper and lower bounds is mainly driven by the change-of-measure penalties~$\frac{r-1}{r}D_r(P_{H^n}\|Q_{H^n})$ and~$D(P_{H^n}\| Q_{H^n})$, which grow approximately linear in $n$. 
This gap results from the suboptimality of the signaling and the change of measure, which are probably not the best choices.

\section{Conclusion}\label{sec:conclusion}

This work presents the first finite-blocklength analysis yielding bounds on the information density that hold uniformly over time along each input–output sequence for time-correlated noncoherent fading channels.
The proposed lower bound directly enables the design of stopping rules with reliability guarantees, making it suitable for VLSF decoding in the non-asymptotic regime.
Promising directions for future research include investigating alternative signaling schemes and reference measures to obtain tighter lower bounds across different channel correlation regimes, as well as developing optimization methods for the free parameters $(r, \sigma_h^2)$.

\IEEEtriggeratref{8}
\bibliographystyle{ieeetr}
\bibliography{reference_LB.bib}

\end{document}